\documentclass[12pt,a4paper,reqno]{amsart}

\usepackage[margin=2.8cm,footskip=1cm]{geometry}
\usepackage[utf8]{inputenc}
\usepackage[T1]{fontenc}
\usepackage[english]{babel}
\usepackage{amsmath}
\usepackage{amsfonts}
\usepackage{amssymb}
\usepackage{amsthm}
\usepackage{ucs}
\usepackage{graphicx}
\usepackage{xcolor}
\usepackage{dsfont}
\usepackage{tikz}
\usepackage{wasysym}
\usepackage{physics}
\usepackage{mathtools}
\usepackage{xifthen}
\usepackage[textwidth=2.5cm,textsize=tiny]{todonotes}
\usepackage{enumitem}
\usepackage{stmaryrd}
\usepackage{constants}
\usepackage[ruled]{algorithm2e}
\usepackage{tikz,pgfplots}
\usepackage[hidelinks]{hyperref}
\usepackage[font={small}]{caption}

\usepackage{constants}

\newconstantfamily{ctrlcst}{
	symbol=\kappa
}
\newconstantfamily{csts}{
	symbol=C
}
\renewconstantfamily{normal}{
	symbol=c
}

\makeatletter
\newcommand{\pushright}[1]{\ifmeasuring@#1\else\omit\hfill$\displaystyle#1$\fi\ignorespaces}
\newcommand{\pushleft}[1]{\ifmeasuring@#1\else\omit$\displaystyle#1$\hfill\fi\ignorespaces}
\makeatother

\renewcommand{\norm}[1]{\|#1\|}

\newcommand{\Z}{\mathbb{Z}}
\newcommand{\R}{\mathbb{R}}

\newcommand{\Zd}{\mathbb{Z}^d}

\newcommand{\Ed}{\mathrm{E}^d}

\newcommand{\e}{\mathbb{E}}
\newcommand{\p}{\mathbb{P}}

\newcommand{\PPP}{\mathcal{T}}

\newcommand{\bbZ}{\mathbb{Z}}

\newcommand{\bbT}{\mathbb{T}}
\newcommand{\bbL}{\mathbb{L}}

\newcommand{\calC}{\mathcal{C}}

\newcommand{\rms}{\mathrm{s}}
\newcommand{\rmt}{\mathrm{t}}

\newcommand{\betac}{\beta_{\mathrm{\scriptscriptstyle c}}}

\newcommand{\FKlaw}{\Phi}

\newcommand{\clusterSet}{\mathrm{cl}}

\theoremstyle{plain}
\newtheorem{theorem}{Theorem}[section]
\newtheorem{lemma}[theorem]{Lemma}

\newtheorem{corollary}[theorem]{Corollary}

\newtheorem{remark}{Remark}[section]

\theoremstyle{definition}

\newtheorem{obs}{Observation}

\author{S\'{e}bastien Ott}
\address{Dipartimento di Matematica e Fisica, Univ. Roma Tre, 00146 Roma, Italy}
\email{ott.sebast@gmail.com}

\title[Weak mixing and analyticity in R.C. and l.t. Ising models]{Weak mixing and analyticity in Random Cluster and low temperature Ising models}

\begin{document}

\begin{abstract}
	In this note we extend the analysis of~\cite{Ott-2019} to the random cluster model. The main result being that the pressure of the finite range ferromagnetic Ising model is analytic as a function of the inverse temperature in the regime \(h=0\), \(\beta>\betac\).
\end{abstract}

\maketitle

%%%%%%%%%%%%%%%%%%%%%%%%%%%%%%%%%%%%%%%%%%%%%%%%%%%%%%%%%%%%%%%%%%%%%%%%%%%%%%%%%%%%%%%%%%%%%%%%%%%

\section{Introduction}

In this note we extend the analysis done for the Ising model in~\cite{Ott-2019} to the Random Cluster model. As a result, we obtain that the pressure (free energy) and local observables of the Random Cluster model are analytic in \(\beta\) whenever the model mixes exponentially fast. The main objective being to prove that the pressure of the finite range ferromagnetic Ising model in dimension \(\geq 3\)  is analytic as a function of the inverse temperature in the regime \(\beta>\betac\).

The treatment follows the one of~\cite{Ott-2019}, relying on mixing of the Glauber dynamic, coupling from the past and cluster expansion. The main difficulty is that the Glauber dynamic for the Random Cluster model is non-local. This additional technicality is dealt with using the results of~\cite{Harel+Spinka-2018} that extend the ones of~\cite{Martinelli+Olivieri-1994} on the Glauber dynamic and a modified coarse-graining argument.

Some technicalities are not repeated here and references to the corresponding arguments/results in~\cite{Ott-2019} are given instead.

The results in dimension \(2\) (for nearest-neighbours models in the LT regime) are known for a fairly long time: duality enable to map the problem to a high-temperature one for which stronger results hold (restricted complete analyticity). See~\cite{Dobrushin+Shlosman-1987,Martinelli+Olivieri+Schonmann-1994,Schonmann+Shlosman-1995,van_Enter+Fernandez+Schonmann+Shlosman-1997} for more details.

\section{Random cluster, Potts model and main results}

\subsection{Some notations}

We denote \(\Ed\) the set of nearest neighbour pairs in \(\Zd\) and will work mainly with the graph \((\Zd,\Ed)\). We will denote \(\bbT_N = \{-N,\cdots,N\}^d\) the \(d\)-dimensional torus that we shall often see as a subset of \(\Zd\), and \(E_{\bbT_N}\) the nearest neighbour edges of \(\bbT_N\).

For \(L\) and \(N\) such that \(2N+1\) is divisible by \(2L+1\), we define \(\bbT_N^L\) to be the subset of \(\bbT_N\) given by
\begin{equation}
	\bbT_N^L = ((2L+1)\Z)^d\cap \bbT_N.
\end{equation}
We will also see \(\bbT_N^L\) as a graph with edges between sites at distance \(2L+1\) (for the graph distance on \(\bbT_N\)). We say that two points in \(\bbT_N^L\) are \emph{connected} in a given \(F\subset\bbT_N^L\) if they are in the graph \(F\) and \emph{star-connected} (or connected \emph{diagonally}) if there exists a sequence of vertices in \(F\) linking them such that any two consecutive vertices are at distance at most \(\sqrt{2}(2L+1)\) from one another.
We denote \(\Gamma_{N}^L\) the set of connected sub-graphs of \(\bbT_N^L\).

For \(x\in\Zd\) (or \(x\in \bbT_N\)) and \(K\geq 0\), we denote \(B_K(x)=\{y\in\Zd:\ \norm{x-y}_{\infty} \leq K\}\) and similarly for \(\bbT_N\). In particular, we have
\begin{gather*}
	B_L(x) \cap B_L(y) =\varnothing,\ \forall x\neq y\in\bbT_N^L,\\
	\bigcup_{x\in\bbT_N^L} B_L(x) = \bbT_N.
\end{gather*}
For \(x\in\bbT_N^L\), denote
\begin{equation*}
E_L(x) = \big\{ \{i,j\}:\ i\in B_L(x), j\in \{i+\mathrm{e}_k,\ k=1,\cdots, d \} \big\},
\end{equation*}where \((\mathrm{e}_k)_s=\delta_{s,k}\). They partition \(E_{\bbT_N}\). We write \(E_L\equiv E_L(0)\) and \(B_L\equiv B_L(0)\).

For functions \(f:\Zd\to \R\) of \(\Ed\to \R\), and \(A,B\subset \Zd\) (or \(\Ed\)) we denote \(f_{A}\) the restriction of \(f\) to \(A\) and, for \(A\cap B=\varnothing\) and \(f,f'\) two functions, \(f_Af'_B\) is the function from \(A\cup B\to\R\) agreeing with \(f\) on \(A\) and with \(f'\) on \(B\).

\(c,c',c'',\tilde{c},\cdots\) denote constants that are allowed to depend on \(\beta,q,d\). Their values can change from one line to another.

\subsection{Ising and Potts models}

Let \(V\subset \Zd\) be finite. Let \(\eta\in\{-1,1\}^{\Zd}\), \(\beta\geq 0\) and \(h\in\R\). The Ising model on \(V\) at inverse temperature \(\beta\) and magnetic field \(h\) with boundary conditions \(\eta\) is the probability measure on \(\{-1,+1\}^{V}\) given by
\begin{equation*}
	\mu_{V,\beta,h}^{\eta}(\sigma) = \frac{1}{Z_{V,\beta,h}^{\mathrm{Ising},\eta}} \exp(\beta\sum_{\substack{\{i,j\}\subset V\\i\sim j}} \sigma_i\sigma_j + \beta\sum_{\substack{i\in V,j\in V^c\\i\sim j}}\sigma_i\eta_j  + h\sum_{i\in V}\sigma_i),
\end{equation*}where \(Z_{V,\beta,h}^{\mathrm{Ising},\eta}\) is the normalization constant (the partition function).

The Potts model is a generalization of the Ising model. Let \(q\geq 2\) be an integer, \(\eta\in\{1,\cdots,q\}^{\Zd}\) and \(\beta\geq 0\). The Potts model on \(V\) at inverse temperature \(\beta\) with boundary conditions \(\eta\) is the probability measure on \(\{1,\cdots,q\}^{V}\) given by
\begin{equation*}
\mu_{V,\beta,q}^{\eta}(\sigma) = \frac{1}{Z_{V,\beta,q}^{\mathrm{Potts},\eta}} \exp(\beta\sum_{\substack{\{i,j\}\subset V\\i\sim j}} \delta_{\sigma_i,\sigma_j} + \beta\sum_{\substack{i\in V,j\in V^c\\i\sim j}} \delta_{\sigma_i,\eta_j}),
\end{equation*}where \(Z_{V,\beta,q}^{\mathrm{Potts},\eta}\) is the normalization constant (the partition function). The Ising model with \(h=0\) at inverse temperature \(\beta\) is equivalent to the Potts model with \(q=2\) at inverse temperature \(2\beta\). Their partition functions are related via
\begin{equation}
	Z_{V,\beta,0}^{\mathrm{Ising},\eta} = e^{-\beta|E_V|}Z_{V,\beta,2}^{\mathrm{Potts},\eta}
\end{equation}where \(E_V\) is the set of nearest-neighbour edges with at least one endpoint in \(V\) ( and we do the assimilation \(1\mapsto 1\) and \(2\mapsto -1\) for the boundary conditions).
The central objects of the present study are the pressures:
\begin{gather}
	\psi_{\mathrm{Ising}}(\beta,h) := \lim_{V\to\Zd}\frac{1}{|V|}\log(Z_{V,\beta,h}^{\mathrm{Ising},\eta}),\\
	\psi_{\mathrm{Potts},q}(\beta) := \lim_{V\to\Zd}\frac{1}{|V|}\log(Z_{V,\beta,q}^{\mathrm{Potts},\eta}),
\end{gather}where the limits are taken in the sense of van Hove. It can be shown (see~\cite{Friedli+Velenik-2017}) that the limit exists and does not depend on the sequence of volumes \(V\) nor on the boundary conditions. We will thus work with cubic volumes and periodic boundary conditions (models on the torus). We will use the notation \(\mathrm{per}\) to denote periodic boundary conditions.

\subsection{Random cluster model}

First introduce a convention: Let \(G=(V,E)\) be a graph. Let \(F\subset E\). We will systematically identify the graph \((V,F)\), the set of edges \(F\) and the function \(E\to \{0,1\}\) that takes value one if the edge is in \(F\) and \(0\) otherwise. It will thus make sense of saying that a function is a subgraph and vice-versa.

Let \(q\geq 1\) and \(\beta\geq 0\). Let \(\eta\in\{0,1\}^{\Ed}\) be such that \(\eta\setminus F\) has at most one infinite cluster for any \(F\subset \Ed\) finite. The random cluster model on \(F\subset \Ed\) at inverse temperature \(\beta\) with cluster weight \(q\) and boundary conditions \(\eta\) is the probability measure on \(\{0,1\}^{F}\) given by
\begin{equation}
	\FKlaw_{F,\beta,q}^{\eta}(\omega) = \frac{1}{Z^{\mathrm{RC},\eta}_{F,\beta,q}} (e^{\beta}-1)^{|\omega|} q^{\kappa_{F}^{\eta}}(\omega),
\end{equation}where \(|\omega| = \sum_{e\in F}\omega(e)\) and \(\kappa_{F}^{\eta}\) is the number of connected components of the graph \((\Zd,\omega\eta_{F^c})\) intersecting \(V_F\) where \(\omega\eta_{F^c}\) is the set of edges obtain by the union of the edges of \(\omega\) and of the edges of \(\eta\setminus F\) and \(V_F\) is the set of vertices being the endpoint of an edge in \(F\).

As in the Potts and Ising cases, introduce the pressure:
\begin{equation}
	\psi_{\mathrm{RC},q}(\beta) := \lim_{F\to\Ed}\frac{1}{|V_F|}\log(Z_{F,\beta,q}^{\mathrm{RC},\eta}),
\end{equation}as previously (\(|V_F|\) is the set of vertices with an endpoint in \(F\)), the limit is taken in the sense of van Hove and does not depend on the sequence of volumes of boundary conditions. We will thus again use periodic boundary conditions and cubic volumes.
\begin{remark}
	Percolation adepts will notice that this definition gives something trivial for \(q=1\) (Bernoulli percolation). Our main interest lying in the application to Ising and Potts model, this is the correct way of defining the pressure. One can recover classical quantities of Bernoulli percolation as follows (\(p=1-e^{-\beta}\)):
	\begin{itemize}
		\item \(\kappa(p)\), the mean number of cluster per vertex is obtained as the limit of the derivative in \(q\) of \(|B_N|^{-1}\log(Z_{B_N,\beta,q})\) evaluated at \(q=1\) (with \(B_N\) the cubic box of side \(2N+1\).
		\item \(\chi(p)\), the mean size of the cluster of the origin, can be obtain by adding a magnetic field to the picture:
		\begin{equation*}
			Z_{N,\beta,q,h} = \sum_{\omega} (e^{\beta}-1)^{|\omega|} \prod_{C\in\clusterSet(\omega)} (e^{h|C|}-1+q).
		\end{equation*}
		\(\chi(p)\) is then obtained as the limit of the second derivative in \(h\) of \linebreak \(\frac{1}{q-1}\frac{1}{|B_N|} \log(Z_{N,\beta,q,h})\) evaluated at \(h=0,q=1\).
	\end{itemize}
	In particular, our analyticity result gives no useful information about analyticity of classical quantities for Bernoulli percolation. See~\cite{Grimmett-1999} for more details.
\end{remark}

The random cluster model is closely linked to the Potts and Ising models via the Edward-Sokal coupling but the only feature we will use is
\begin{equation}
	Z_{E_{\mathbb{T}_N},\beta,q}^{\mathrm{RC},\mathrm{per}} = Z_{\mathbb{T}_N,\beta,q}^{\mathrm{Potts},\mathrm{per}},
\end{equation}where \(\mathbb{T}_N\) is the \(d\)-dimensional torus with side \(2N+1\), \(\mathbb{T}_N=\{-N,\cdots, N\}^d\) (and \(E_{\mathbb{T}_N}\) is the set of nearest neighbour edges in \(\mathbb{T}_N\)).

A central feature of the random cluster measures is that they have the lattice FKG property (for the canonical partial order on \(\{0,1\}^{\Ed}\)). In particular, if \(1\) (resp. \(0\)) denote the constant configuration \(1\) (resp. \(0\)), one has the stochastic ordering
\begin{equation*}
	\FKlaw_{F,\beta,q}^{0} \preccurlyeq \FKlaw_{F,\beta,q}^{\eta} \preccurlyeq \FKlaw_{F,\beta,q}^{1}.
\end{equation*}

Our main result is conditional to two (natural) hypotheses. Denote \(\Lambda_N\) the set of edges with at least one endpoints at \(\norm{\ }_{\infty}\)-distance \(\leq N\) from \(0\). Our first hypotheses is mixing:
\begin{enumerate}[label=(H\arabic*)]
	\item \label{hyp:exp_weak_mix} Exponential weak mixing: for any \(a>1\), there exist \(c>0,N_0\geq 0\) such that for any \(N\geq N_0\),
	\begin{equation*}
	d_{\mathrm{TV}}\Big(\FKlaw_{\Lambda_{aN},\beta,q}^{0}|_{\Lambda_N},\FKlaw_{\Lambda_{aN},\beta,q}^{1}|_{\Lambda_N}\Big)\leq e^{-cN},
	\end{equation*}where \(\FKlaw_{\Lambda_{aN},\beta,q}^{*}|_{\Lambda_N}\) is the random cluster measure \(\FKlaw_{\Lambda_{aN},\beta,q}^{*}\) restricted to \(\Lambda_N\) and \(d_{\mathrm{TV}}\) is the total variation distance.
\end{enumerate}
The second hypotheses is more specific to the Random Cluster model and it should be unnecessary in the derivation of Theorem~\ref{thm:main} (beside being used in the proof of the first one). It declines in two versions corresponding to the regimes of non-percolation and percolation.

\begin{enumerate}[resume, label=(H\arabic*)]
	\item \label{hyp:exp_dec_finite_connexions} Exponential decay of finite connexions. One of the two following occurs
	\begin{enumerate}[label=(H2.\arabic*)]
		\item For any \(a>1\), there exist \(c>0,N_0\geq 0\) such that for any \(N\geq N_0\), \label{hyp:non_perco}\begin{equation*}
			\FKlaw_{B_{aN},\beta,q}^{1}\big( A_N \big) \leq e^{-cN},
		\end{equation*}
		where \(A_N\) is the event that \(B_{N}\) contains a cluster of diameter larger than \(L/100\).
		\item For any \(a>1\), there exist \(c>0,N_0\geq 0\) such that for any \(N\geq N_0\), \label{hyp:perco}\begin{equation*}
		\sup_{\eta} \FKlaw_{B_{aN},\beta,q}^{\eta}\big(A'_N\big) \geq 1- e^{-cN},
		\end{equation*}where \(A'_N\) is the event that \(B_N\) contains a cluster connecting all sides of \(B_N\) and that the second largest cluster in \(B_N\) is of radius at most \(N/100\).
	\end{enumerate}
\end{enumerate}

\subsection{Results}

The results are stated and the proofs conducted for the nearest-neighbour models but both extend to finite range models.

Our main result is
\begin{theorem}
	\label{thm:main}
	Let \(q,\beta\) be such that \(\FKlaw_{\beta,q}\) satisfies~\ref{hyp:exp_dec_finite_connexions} and~\ref{hyp:exp_weak_mix}. Then, there exists \(\epsilon>0\) such that the function
	\begin{equation*}
		z\mapsto \psi_{\mathrm{RC},q}(z)
	\end{equation*}is analytic in the domain \(\{z\in\mathbb{C}:\ |z-\beta|<\epsilon\}\).
\end{theorem}

From Theorem~\ref{thm:main}, one can deduce the main motivation of this note
\begin{corollary}
	Assume \(d\geq 3\). For any \(\beta>\betac\), there exists \(\epsilon>0\) such that
	\begin{equation*}
		z\mapsto \psi_{\mathrm{Ising}}(z,0)
	\end{equation*}
	is analytic in the domain \(\{z\in\mathbb{C}:\ |z-\beta|<\epsilon\}\).
\end{corollary}
\begin{proof}
	Hypotheses~\ref{hyp:exp_weak_mix} for \(q=2,\beta>\betac, d\geq 3\) is the main result of~\cite{Duminil-Copin+Goswami+Raoufi-2019}. Hypotheses~\ref{hyp:perco} follows from the validity of Pisztora's coarse-graining,~\cite{Pisztora-1996}, which is a consequence of~\cite{Bodineau-2005} for \(q=2,\beta>\betac\). Using the correspondence between the partition functions of the Ising and Random Cluster models with \(q=2\) gives the result in \(d\geq 3\).
\end{proof}
As mentioned in the introduction, more is known for nearest-neighbour models in dimension \(2\) via planar duality.

\begin{corollary}
	For any \(q\geq 1\), and any \(\beta<\betac(q)\), there exists \(\epsilon>0\) such that
	\begin{equation*}
	z\mapsto \psi_{\mathrm{RC},q}(z)
	\end{equation*}
	is analytic in the domain \(\{z\in\mathbb{C}:\ |z-\beta|<\epsilon\}\). In particular, the same holds for \(\psi_{\mathrm{Potts},q}(z)\) and \(q\geq 2\) integer.
\end{corollary}
\begin{proof}
	Hypotheses~\ref{hyp:non_perco} and~\ref{hyp:exp_weak_mix} are verified whenever the model exhibit exponential decay of connections probabilities uniformly over boundary conditions (simple use of the lattice FKG property). The later was proven (in any dimensions) in the whole regime \(\beta<\betac\) for \(q=2\) in~\cite{Aizenman+Barsky+Fernandez-1987}, for \(q=1\) in~\cite{Aizenman+Barsky-1987,Menshikov-1986} (see also~\cite{Duminil-Copin+Tassion-2016} for an alternative proof of \(q=1,2\)) and for \(q\geq 1\) in~\cite{Duminil-Copin+Raoufi+Tassion-2017}.
\end{proof}

To state our last result, we need additional notation. For any \(A\subset \Ed \) let \(g_A(\omega) = \prod_{e\in A} \omega_e\). Let \(R<\infty\) and \(W\) be a complex function depending only on the edges at distance \(\leq R\) of \(0\). Denote \(W_x\) the translate of \(W\) by \(x\). Define then
\begin{gather*}
	Z_{F,\beta,q}^{\mathrm{RC},\eta}(\lambda W) = \sum_{\omega\subset F}(e^{\beta}-1)^{|\omega|} q^{\kappa_F^{\eta}(\omega)} e^{\lambda\sum_{x\in V_F} W_x(\omega)}\\
	\tilde{\psi}_{\mathrm{RC},q}(\beta,\lambda W) := \lim_{F\to\Ed}\frac{1}{|V_F|}\log(Z_{F,\beta,q}^{\mathrm{RC},\eta}),
\end{gather*}where the limit is taken in the sense of van Hove. It is not clear that the limit even exists (if it does, the classical arguments used for the existence of the pressure ensure that it does not depend on the sequence of volumes nor on the boundary conditions).

\begin{theorem}
	\label{thm:Correlation_analy}
	Let \(q,\beta\) be such that \(\FKlaw_{\beta,q}\) satisfies~\ref{hyp:exp_dec_finite_connexions} and~\ref{hyp:exp_weak_mix}. Then, for any \(A\subset \Ed\) finite, there exists \(\epsilon>0\) such that the function
	\begin{equation*}
	z\mapsto \FKlaw_{z,q}(g_A)
	\end{equation*}is analytic in the domain \(\{z\in\mathbb{C}:\ |z-\beta|<\epsilon\}\).
\end{theorem}
\begin{theorem}
	\label{thm:GenPressure_analy}
	Let \(q,\beta\) be such that \(\FKlaw_{\beta,q}\) satisfies~\ref{hyp:exp_dec_finite_connexions} and~\ref{hyp:exp_weak_mix}. Then, for any \(R<\infty\) and \(W\) complex function depending only on edges at distance at most \(R\) from \(0\) with \(\norm{W}_{\infty}\leq 1\), the function
	\begin{equation*}
		\lambda\mapsto \tilde{\psi}_{\mathrm{RC},q}(\beta,\lambda W)
	\end{equation*}is well defined and is analytic in a neighbourhood of \(0\).
\end{theorem}

Theorem~\ref{thm:GenPressure_analy} is a straightforward generalization of Theorem~\ref{thm:main}. We do not present its detailed proof as it is a simple adaptation of the proof of Theorem~\ref{thm:main} (only section~\ref{section:polymer_rep_convergence_cluster_expansion} needs minor changes that are only notationally heavier).

The proof of Theorem~\ref{thm:Correlation_analy} is a standard adaptation of the one of Theorem~\ref{thm:main}. We only sketch the modifications in section~\ref{sec:sketch_Corr_analy} and leave the details to the interested reader.

\subsection{Comments and related problems}

The present study close the question of analyticity of the pressure for the finite range ferromagnetic Ising model on \(\Zd\). An interesting question is the extension of the results of~\cite{Ott-2019} and of the present work to the Ising model with infinite range interactions. This seems doable by the methods exposed here and in~\cite{Ott-2019} and a modified polymer expansion (see section~\ref{section:polymer_rep_convergence_cluster_expansion}).

On the purely technical side, by carefully recording the dependency of all quantities in \(\beta\) and using monotonicity of Random Cluster measure, one should be able to get a quantitative version of Theorem~\ref{thm:main} which should give an ``explicit'' region of analyticity.

Another related question is the one of analytic continuation: take the Ising model with \(\beta>\betac\). While the results of Isakov~\cite{Isakov-1984}, prevent the extension of \(h\mapsto\psi(\beta,h)\) through \(h=0\), one can ask whether such analytic continuation is possible through a purely imaginary value of the field (the pressure being analytic when the real part of \(h\) is either \(>0\) or \(<0\) by Lee-Yang Theorem).

Another interesting direction is the following: one can notice many similarity in the ways~\cite{Harel+Spinka-2018} construct finitary coding of lattice-FKG random fields and the construction involved in Theorem~\ref{thm:dep_encoding_coupling} (namely: Glauber dynamic and coupling from the past). It would be interesting to formulate and prove some general equivalence statement between existence of finitary coding with certain bounds on the coding cluster volume and analyticity properties of the pressure (which is in a sense what is done in the present paper and~\cite{Ott-2019} for the Random Cluster and Ising models).

Finally, a deep study of the relationship between strong mixing properties and strong analyticity properties was realized by Dobrushin and Shlosman~\cite{Dobrushin+Shlosman-1987}. One can hope for a similar generic study of the relationship between weak mixing and ``soft'' analytic properties (analyticity of bulk quantities). We plan to come back to this issue in a near future.

\section{Glauber dynamic}

\subsection{Graphical representation of Glauber dynamic}

We introduce directly the graphical representation of the Glauber dynamic introduced in~\cite{Schonmann-1994} (which couples the dynamics in all volumes and with all boundary conditions). We refer to~\cite{Martinelli-1999} and to the original paper for details and proof of the unproven statements.

To each edge \(e\in \Ed\), attach a copy of \(\R_-\) with a rate one Poisson point process on it, denoted \(\PPP_e=\{t_{e,1}>t_{e,2}>\cdots\}\). To each point \(t_{e,k}\) of the Poisson point processes, attach a uniform random variable \(U_{e,k}\) on \([0,1]\). All this such that the family \(\{\PPP_e:\ e\in \Ed\}\cup\{U_{e,k}:\ e\in \Ed,k\geq 1\}\) forms an independent family whose law and expectation will be denoted \(\p,\e\). Almost surely, all \(t_{e,k}\)s are distinct, they can thus be ordered in increasing order. For \(\Lambda\subset \Ed\) finite, we will denote \(\PPP_{\Lambda}\) the superposition of \(\PPP_e,e\in\Lambda\).

For an edge \(e=\{i,j\}\) and a configuration \(\omega\in\Omega=\{0,1\}^{\Ed}\) with at most one infinite cluster, define the update rates
\begin{equation*}
p_e(\omega) = \begin{cases}
1-e^{-\beta} & \text{ if } i\xleftrightarrow{\omega\setminus e} j,\\
\frac{e^{\beta} -1}{e^{\beta} -1 + q} & \text{ else.}
\end{cases}
\end{equation*}

Let now \(\Lambda\subset \Ed\) be finite and \(\bar{\eta} = \eta_t\in\Omega^{\R_-}\) be a fixed collection of boundary conditions with at most one infinite cluster in \(\eta_t\setminus \Lambda\). Let \(\omega\in\Omega\) be a fixed starting configuration.

We can now define \(\sigma_{t,\Lambda}^{\omega,\bar{\eta}}\) the configuration after time \(t\) sampled using Glauber dynamic in volume \(\Lambda\) with (evolving) boundary conditions \(\bar{\eta}\) as follows:
\begin{itemize}
	\item Let \(-t<t_{1}<t_{2}<t_{3}<\cdots<t_{M}\leq 0\) be the ordered sequence of arrival times of the collection \((\PPP_e)_{e\in\Lambda}\) in the interval \([-t,0]\) (as \(\Lambda\) and \(t\) are finite, it is almost surely a finite sequence). Denote also \(e_1,e_2,\cdots\) the edges associated to \(t_1,t_2,\cdots\) and \(U_1,U_2,\cdots\) the uniform random variables associated.
	\item Construct \(\sigma_{t,\Lambda}^{\omega,\bar{\eta}}\) as follows:
	\begin{enumerate}
		\item Set \(\sigma^{(0)} = \omega_{\Lambda}\).
		\item For all \(1\leq k\leq M\), construct \(\sigma^{(k)}\) from \(\sigma^{(k-1)}\) by setting
		\begin{equation*}
			\sigma^{(k)}(e) = \begin{cases}
			\mathds{1}_{U_k<p_{e_k}(\sigma^{(k-1)}\eta_t|_{\Lambda^c} )} & \text{ if } e= e_k,\\
			\sigma^{(k-1)}(e) & \text{else.}
			\end{cases}
		\end{equation*}
		\item Set \(\sigma_{t,\Lambda}^{\omega,\bar{\eta}}=\sigma^{(M)}\).
	\end{enumerate}
\end{itemize}

One can notice that for fixed \(t,\Lambda\) and a given realization of the P.P.P.s and of the uniforms, the map \((\bar{\eta},\omega)\mapsto \sigma_{t,\Lambda}^{\omega,\bar{\eta}}\) is deterministic. Moreover, for a fixed \(t,\Lambda\) and for a fixed realization of the P.P.P.s and of the uniform, \(\sigma_{t,\Lambda}^{\omega,\bar{\eta}}\) is non-decreasing in \(\omega\) and \(\bar{\eta}\) (as a simple consequence of the lattice FKG property of the Random Cluster Measure).

We will also need the ``intermediate'' steps of construction: for \(0\leq s\leq t\), define \(\sigma_{t,s,\Lambda}^{\omega,\bar{\eta}}\) in the same way as we defined \(\sigma_{t,\Lambda}^{\omega,\bar{\eta}}\) but using only the updates up to time \(-t+s\) (i.e.: replacing \(t_1<t_2<\cdots<t_M\leq 0\) by \(t_1<t_2<\cdots<t_{M_s}\leq -t+s \)). In particular, \(\sigma_{t,\Lambda}^{\omega,\bar{\eta}} = \sigma_{t,t,\Lambda}^{\omega,\bar{\eta}}\) and \( \sigma_{t,0,\Lambda}^{\omega,\bar{\eta}} = \omega\). We will denote \(1\) (resp. \(0\)) the constant configuration \(1\) (resp. \(0\)) and do the same for the sequence of configurations \(\bar{\eta}\). The ordering mentioned above implies that for any \(\omega,\bar{\eta}\),
\begin{equation*}
	\sigma_{t,s,\Lambda}^{0,0} \leq \sigma_{t,s,\Lambda}^{\omega,\bar{\eta}} \leq \sigma_{t,s,\Lambda}^{1,1}.
\end{equation*}
For \(\Delta\subset\Lambda\), we will denote \(\sigma_{t,s,\Lambda,\Delta}^{\omega,\bar{\eta}}\) the restriction of \(\sigma_{t,s,\Lambda}^{\omega,\bar{\eta}}\) to \(\Delta\). We will omit volumes from notation when considering the infinite volume dynamic and write \(\sigma_{t,s}^{\omega}(\Delta)\) for the configuration in \(\Delta\).

The interest of this procedure is that for any finite volume \(\Lambda\), any (non-evolving) boundary condition \(\bar{\eta}_t=\eta\) and any starting configuration \(\omega\), the law of \(\sigma_{t,\Lambda}^{\omega,\eta}\) converges as \(t\to\infty\) to \(\FKlaw_{\Lambda,\beta,q}^{\eta}\). Moreover, if \(\sigma_{t,s,\Lambda,\Delta}^{0,0}=\sigma_{t,s,\Lambda,\Delta}^{1,1}\) then \(\sigma_{t,s,\Lambda,\Delta}^{0,0}\) is distributed according to the \(\Delta\)-marginal of \(\FKlaw_{\beta,q}\) and is determined by the restriction of the PPP and of the uniforms inside \(\Lambda\times [-t,0]\). When there exists a unique infinite volume random cluster measure, the dynamic also converges to this unique measure (in the sense that the law of the restriction of \(\sigma_{t}^{\omega}\) to any finite volume converges to the associated marginal of \(\FKlaw_{\beta,q}\)). In the case there is a unique infinite volume measure, we will denote \(\sigma(e,t)\) the state of the edge \(e\) at time \(t\):
\begin{equation*}
	\sigma(e,t) = \lim_{s\to\infty} \sigma_{t+s,s}^{\omega}(e).
\end{equation*}
The limit is a.s. well defined and does not depend on \(\omega\) (again, in the uniqueness regime).

A key input to our analysis is

\begin{theorem}
	\label{thm:Glauber_point_to_box_connection}
	Suppose that \(\FKlaw\) has the exponential weak spatial mixing property. Then, there exist \(c>0\), \(\alpha<\infty\), and \( N_0\geq 0\) such that for any \(N\geq N_0\), and \(e = \{0,\mathrm{e}_i\}\),
	\begin{equation}
		\p\big(\sigma_{\alpha N,E_{N}}^{0,0}(e)\neq \sigma_{\alpha N,E_{N}}^{1,1}(e)\big) \leq e^{-cN}.
	\end{equation}
\end{theorem}

This proof follows closely the one of~\cite{Martinelli+Olivieri-1994} for the Ising model and the proof of a more general statement can be found in~\cite{Harel+Spinka-2018}.

\section{Coupling construction and properties}

The goal of this section is the proof of
\begin{theorem}
	\label{thm:dep_encoding_coupling}
	Suppose that~\ref{hyp:exp_weak_mix} and~\ref{hyp:exp_dec_finite_connexions} are satisfied. Then there exists \(L_0\geq 0\) such that for any \(L\geq L_0\), one can construct a probability measure \(P\) on \((\Gamma_{N}^L)^{\bbT_N^L} \times \{0,1\}^{E_{\bbT_N}}\) with the following properties: let \((C,\omega)\sim P\),
	\begin{itemize}
		\item \(\omega\sim \FKlaw_{\bbT_N,\beta,q}\).
		\item For all \(x\in \bbT_N^L\), \(C_x\ni x\).
		\item For any \(\Delta_1,\Delta_2\subset \bbT_N^L\), \(f,g\) functions supported on \(\Delta_1,\Delta_2\) respectively, and \(\Delta_i \subset D_i\subset \bbT_N^L, i=1,2\) with \(D_1\cap D_2=\varnothing\)
		\begin{equation}
		\label{eq:decoupling}
			P(fg\mathds{1}_{C_{\Delta_1}=D_1}\mathds{1}_{C_{\Delta_2}=D_2}) = P(f\mathds{1}_{C_{\Delta_1}=D_1})P(g\mathds{1}_{C_{\Delta_2}=D_2})
		\end{equation}where \(C_{\Delta} = \bigcup_{x\in \Delta} C_x\).
		\item There exist \(c>0,c'\geq 0\) such that for any \(x\in \bbT_N^L\)
		\begin{equation}
		\label{eq:exp_dec_dep_clusters}
			P(|C_x|\geq l)\leq c'e^{-cl},
		\end{equation}moreover, \(c\to\infty\) as \(L\to\infty\).
	\end{itemize}
\end{theorem}

It will be convenient to re-formulate the last item in the following form
\begin{corollary}
	\label{cor:exp_dec_cluster_set}
	With the same hypotheses and notations as Theorem~\ref{thm:dep_encoding_coupling}, there exists \(a\geq 0\) such that for any \(\Delta\subset\bbT_N^L\),
	\begin{equation}
		P(|C_{\Delta}| \geq l) \leq e^{-cl}e^{a |\Delta|}.
	\end{equation}
\end{corollary}
\begin{proof}
	See~\cite[Lemma 4.4]{Ott-2019}.
\end{proof}

\begin{remark}
	Be careful, \(a\) depends on \(L\). The proof gives \(c= O(L)\) but we will need only the weaker statement.
\end{remark}

\subsection{Information clusters}

As in~\cite{Ott-2019}, the centre of the present analysis is a coarse graining procedure of the ``information cluster'' of space-time regions in the graphical representation of Glauber dynamic recalled in the previous section. Let \(K>0\) be a positive real (we will take \(K=\alpha L\) for some fixed \(\alpha\) sufficiently large) and \(L\geq 0\) be an integer. For \(N\) such that \(2N+1\) is divisible by \(2L+1\), define the ``coarse-grained space-time'': \(\bbL_N^{L,K}\) to be the subset of \(\bbT_N\times \R_{+}\)
\begin{equation*}
\bbL_N^{L,K}  = \bbT_N^L\times K\bbZ_+.
\end{equation*}We will denote \(x=(x^{\rms},x^{\rmt})\in\bbL_{N}^L\) with \(x^{\rms}=(x^{\rms}_1,\cdots,x^{\rms}_d)\in \bbT_N^L\) and \(x^{\rmt}\in K\bbZ_+\). In particular, the semi-open boxes
\begin{equation*}
E_{L,K}(x) =  E_L(x^{\rms})\times[x^{\rmt},x^{\rmt}+K),
\end{equation*}partition the space \(E_{\bbT_N}\times\R_+\). We equip \(\bbL_N^{L,K}\) with a graph structure by putting an edge between \(x\) and \(y\) if one of the two following occurs
\begin{itemize}
	\item \(x^{\rmt} = y^{\rmt}\) and \(x^{\rms},y^{\rms}\) are neighbours in \(\bbT_N^L\).
	\item \(x^{\rms}=y^{\rms}\) and \(|x^{\rmt} - y^{\rmt}| = K\).
\end{itemize}
We will say that a point \(x\in \bbL_N^{L,K}\) is \emph{open-good} if
\begin{itemize}
	\item \(\sigma_{x^{\rmt}+3K/2,s,E_{2L}(x^{\rms}),E_{3L/2}(x^{\rms})}^{0,0} = \sigma_{x^{\rmt}+3K/2,s,E_{2L}(x^{\rms}),E_{3L/2}(x^{\rms})}^{1,1}\) for all \(s\in[K/2, 3K/2]\). In words: the configuration at all times between \(-x^{\rmt} - 3 K/2\) and \( -K/2\) in \(E_{3L/2}(x^{\rms})\) is sampled uniformly over boundary conditions on the outside of \(E_{2L}(x^{\rms})\) and initial conditions at time earlier than \(-x^{\rmt}-3K/2\). In other words: for all \(s\in[x^{\rmt}- 3 K/2, x^{\rmt}]\) and \(e\in E_{3L/2}(x^{\rms})\), \(\sigma(e,s )\) does not depend on the PPP outside \(E_{2L,3K/2}(x)\) nor on the associated uniforms.
	\item For all \(s\in[K/2, 3K/2]\), \(\sigma_{x^{\rmt}+3K/2,s,E_{2L}(x^{\rms}),E_{3L/2}(x^{\rms})}^{0,0}\) contains a cluster connecting all sides of \(E_{3L/2}(x^{\rms})\) and at most one cluster of diameter greater or equal to \(L/100\).
\end{itemize}

We will say that a point \(x\in \bbL_N^{L,K}\) is \emph{close-good} if it satisfies the first condition above and for all \(s\in[K/2, 3K/2]\), \(\sigma_{x^{\rmt}+3K/2,s,B_{2L}(x^{\rms}),B_{3L/2}(x^{\rms})}^{0,0}\) contains no cluster of diameter greater or equal to \(L/100\).

Now, explore the properties of good boxes. The idea is that the uniform sampling part guaranties locality of the information needed to construct the configurations (as in~\cite{Ott-2019} for the Ising model) while the largest cluster part forces locality of the information needed to construct the configuration inside a ``surface'' of good blocs (we will use open-good boxes for \(\beta>\betac\) and close ones for \(\beta<\betac\)). We now make the locality statement precise in the next two lemmas. The second one is the main difference between the procedure of~\cite{Ott-2019} and the present one: it is the property handling the non-locality of the Glauber dynamic for the Random Cluster model.
\begin{lemma}
	Let \(C\subset\bbL_N^{L,K}\) be finite. Suppose that all \(x\in C\) are good (open or close). Then, for all \((e,t)\in \bigcup_{x\in C} E_{3L/2,K}(x)  \), \(\sigma(e,t)\) does not depend on the PPP and the uniforms outside \(\bigcup_{x\in C} E_{2L,3K/2}(x)\).
\end{lemma}
\begin{proof}
	If \(x\) is a good point, the state of an edge \(e\in E_{3L/2}(x^{\rms})\) at time \(t\in[-x^{\rmt}-K,-x^{\rmt}]\) does not depend on the PPP and uniforms outside of \(E_{2L,3K/2}\). Inclusion of event give the result.
\end{proof}

We say that \(S\subset\bbL_N^{L,K}\) finite is a \emph{decoupling surface} if
\begin{itemize}
	\item \(S\) is connected.
	\item Either all \(x\in S\) are open-good or all \(x\in S\) are close-good.
	\item \(\bbL_N^{L,K}\setminus S\) is composed of star-connected components, exactly one infinite, and some finite, whose union is denoted \(A(S)\) (possibly empty), such that \(S\cup A(S)\) is connected.
\end{itemize}

We denote
\begin{gather*}
S'= \bigcup_{x\in S} E_{3L/2,K}(x),\quad S'_t=\{(e,t)\in S' \},\quad \bar{S}' = \bigcup_{x\in S} E_{2L,3K/2}(x)\\
\mathring{S} = \bigcup_{x\in S\cup A(S)} E_{L,K}(x),\quad \mathring{S}_t=\{(e,t)\in\mathring{S} \},\\
\bar{S} = \bigcup_{x\in S\cup A(S)} E_{2L,3K/2}(x).
\end{gather*}

\begin{lemma}
	Suppose \(S\) is a decoupling surface. Then, the state of an edge \(e\) at time \(-t\) such that \((e,t)\in \mathring{S}\) does not depend on the PPP and uniforms outside of \(\bar{S}\).
\end{lemma}
\begin{proof}
	Fix a realization of the PPP and the uniforms such that \(S\) is a decoupling surface and let \((e,t)\in \mathring{S}\). If \((e,t)\in S'\), there is noting to check by the first property of good sites.
	
	Otherwise, one only needs to check that the update rate \(p_e(\omega)\) for the edge \(e\) at time \(-t\) does not depend on on the PPP and uniforms outside of \(\bar{S}\). The second property of good sited implies that at time \(-t\) the fact that the two endpoints of \(e\) are connected is determined by the state of the edges inside \(\mathring{S}_t\) (the configuration at time \(-t\) contains an annulus in \(S'_t\) surrounding \(e\) crossed by at most one cluster). Moreover, the state of the edges in \(S'\) is determined by the PPP and uniforms inside \(\bar{S}'\) by the first property of good sites. The state of \((e,t)\) is therefore independent of the PPP and uniforms outside of \(\bar{S}\).
\end{proof}

\subsection{Comparison with Bernoulli percolation}

From now on, we work with \(K=\alpha L\), \(\alpha\) taken large enough (e.g. as in Theorem~\ref{thm:Glauber_point_to_box_connection}).

\begin{lemma}
	\label{lem:good_close_whp}
	If~\ref{hyp:exp_weak_mix} and~\ref{hyp:non_perco} are satisfied, there exists \(c>0,L_0\geq 0\) such that for any \(L\geq L_0\),
	\begin{equation}
		\p(0 \textnormal{ is good-close}) \geq 1-e^{-cL}.
	\end{equation}
\end{lemma}
\begin{proof}
	\begin{multline*}
		\p(0 \textnormal{ is not good-close}) \leq \p(\exists (t,e)\in[0,K]\times E_{3L/2}: \ \sigma_{t,E_{2L}}^{0,0}(e) \neq \sigma_{t,E_{2L}}^{1,1}(e)) +\\
		+ \p(\{\cdots\}, \exists s\in[K/2,3K/2]:\ \sigma_{3K/2,s,E_{2L},E_{3L/2}}^{0,0}\in A_L)
	\end{multline*}
	where \(\{\cdots\}\) means \(\{\forall s\in[K/2,3K/2]: \ \sigma_{3K/2,s,E_{2L},E_{3L/2}}^{0,0} = \sigma_{3K/2,s,E_{2L},E_{3L/2}}^{1,1}\}\), and \(A_L\) is the event that there exists a cluster in \(E_{3L/2}\) of diameter at least \(L/100\). The first term is bounded from above by \(c'L^{d+1}e^{-\tilde{c}L}\) for some \(c'\geq 0,\tilde{c}>0\) by Theorem~\ref{thm:Glauber_point_to_box_connection} and a union bound. The second term is upper bounded by \(c''L^{d+1}e^{-\tilde{c}'L}\) for some \(c''\geq 0,\tilde{c}'>0\) by hypotheses~\ref{hyp:non_perco} a union bound. Implementation of the union bounds is the same as the one in the proof of~\cite[Lemma 4.3]{Ott-2019}.
\end{proof}

\begin{lemma}
	\label{lem:good_open_whp}
	If~\ref{hyp:exp_weak_mix} and~\ref{hyp:perco} are satisfied, there exists \(c>0,L_0\geq 0\) such that for any \(L\geq L_0\),
	\begin{equation}
	\p(0 \textnormal{ is good-open}) \geq 1-e^{-cL}.
	\end{equation}
\end{lemma}
\begin{proof}
	The proof is the same as the one of Lemma~\ref{lem:good_close_whp} using~\ref{hyp:perco} instead of~\ref{hyp:non_perco}.
\end{proof}

Now, notice that \(x\) being good-close or good-open depends only on the randomness (PPP and uniforms) inside \(E_{2L,3K/2}(x)\). In particular, the state of one site \(x\in \bbL_N^{L,K}\) is independent of the state of the sites that are not nearest-neighbours or diagonal neighbours. This allows one to use~\cite[Theorem 1.3]{Liggett+Schonmann+Stacey-1997} to prove
\begin{lemma}
	\label{lem:domination_by_Bernoulli}
	Suppose that hypotheses~\ref{hyp:exp_weak_mix} and~\ref{hyp:non_perco} hold. Then, for any \(L\), there exists \(p_L\in[0,1]\) such that
	\begin{itemize}
		\item the set of sites that are not closed-good is dominated by a site Bernoulli percolation of parameter \(p_L\),
		\item \(p_L\to 0\) when \(L\to\infty\).
	\end{itemize}
	The same statement holds with~\ref{hyp:non_perco} replaced by~\ref{hyp:perco} and ``closed-good'' by ``open good''.
\end{lemma}

\begin{remark}
	One gets the quantitative bound \(p_L\leq e^{-c'L}\) for some \(c'>0\).
\end{remark}

\subsection{Information cluster and radius to volume bound}

Fix \(L\geq 0\). Let \(\calC'\) be the set of sites in \(\bbL_N^L\), that are at distance at most \(2\) from a point connected to \(0\) by a path of bad sites. Notice that \(\calC'\) contains necessarily a decoupling surface surrounding \(E_{L,K}\). Let \(\calC = \bigcup_{x\in \calC'} \{x^{\rms}\}\). The goal of this section is to prove
\begin{lemma}
	\label{lem:rad_to_vol}
	Suppose hypotheses~\ref{hyp:exp_weak_mix} and~\ref{hyp:exp_dec_finite_connexions} are satisfied. Then, there exist \(c>0,c'\geq 0\) such that
	\begin{equation}
		\p(|\calC|\geq k)\leq c'e^{-ck}.
	\end{equation}
	Moreover, \(c\to\infty\) as \(L\to\infty\).
\end{lemma}
\begin{proof}
	Suppose~\ref{hyp:non_perco} holds (the same procedure works for~\ref{hyp:perco}). First \(|\calC|\leq|\calC'|\). We say that two points \(x,y\in\bbL_N^{L,K}\) are connected if there exists a sequence of sites starting at \(x\) and ending at \(y\), that are not good-close, and such that any two consecutive sites in the sequence are a distance at most \(2\) (for the graph distance on \(\bbL_N^{L,K}\)). Choose \(L\) such that the \(p_L\) obtained via Lemma~\ref{lem:domination_by_Bernoulli} is \(\leq (100d)^{-2}\). Write \(C_0\) the set of sites that are connected to \(0\). Standard arguments (Peierls argument) for very small \(p\) Bernoulli percolation imply the existence of \(c>0,c'\geq 0\) such that
	\begin{equation*}
		\p(|C_0|\geq k)\leq c'e^{-ck}.
	\end{equation*}
	As \(\calC'\subset C_0\), this concludes the proof. Taking \(L\) larger allows to take \(p_L\) as small as wanted and thus \(c\) as large as wanted.
\end{proof}

\subsection{Gathering the pieces, proof of Theorem~\ref{thm:dep_encoding_coupling}}

Consider the measure \(\p\). Set \(\omega = \lim_{t\to\infty} \sigma_{t,B_N}^{1,\mathrm{per}}\). Then, let \(\tilde{C}_x\) be the set of sites (in \(\bbL_N^L\), containing \(x\)) that are at distance at most \(2\) from a point connected to \(x\) by a path of bad sites. Define \(C_x=\bigcup_{y\in \tilde{C}_x} \{y^{\rms}\} \). \(\omega\sim \FKlaw_{N,\beta,q}^{\mathrm{per}}\) by convergence of the dynamic in finite volume. \eqref{eq:decoupling} is a direct consequence of the decoupling property of decoupling surfaces, of the product form of \(\p\) and of the dependency of the event \(\{C_x=C\}\) on the PPPs and uniforms attached to sites of \(C\) only. \eqref{eq:exp_dec_dep_clusters} is then Lemma~\ref{lem:rad_to_vol}.

\section{Proof of the main result}
\label{section:polymer_rep_convergence_cluster_expansion}

We present here the proof of Theorem~\ref{thm:main} conditionally on Theorem~\ref{thm:dep_encoding_coupling}. For this section, we fix \(q\geq 1\) and omit it from notation. We also fix \(\beta\) such that hypotheses~\ref{hyp:exp_weak_mix} and~\ref{hyp:exp_dec_finite_connexions} are verified for \(\FKlaw_{\beta,q}\). As we are working only with the Random Cluster model in this section, we will also drop the associate dependency in the notations. Most arguments in this section are very close to the ones presented in~\cite{Ott-2019}. They are nevertheless repeated with the hope to improve their presentation.

Recall that the pressure is
\begin{equation*}
\psi(\beta) = \lim_{N\to\infty}\frac{1}{|\mathbb{T}_N|}\log(Z_{\mathbb{T}_N,\beta}^{\mathrm{per}}).
\end{equation*}
We want to study it in a complex neighbourhood of \(\beta\). We are thus interested in the limit as \(N\to\infty\) of
\begin{equation*}
\frac{1}{|\mathbb{T}_N|}\log(Z_{\mathbb{T}_N,\beta+z}^{\mathrm{per}}) = \frac{1}{|\mathbb{T}_N|}\Big[\log(Z_{\mathbb{T}_N,\beta}^{\mathrm{per}}) + \log(\FKlaw_{N,\beta}^{\mathrm{per}}\Big( \frac{(e^{\beta+z}-1)^{|\omega|}}{(e^{\beta}-1)^{|\omega|}}\Big))\Big].
\end{equation*}
As we know that the first term in the RHS converges to \(\psi(\beta)\), we need to study the existence and analyticity of
\begin{equation*}
F_{\beta}(z) := \lim_{N\to\infty} F_{N,\beta}(z) := \lim_{N\to\infty} \frac{1}{|\mathbb{T}_N|}\log(\FKlaw_{N,\beta}^{\mathrm{per}}\Big( \frac{(e^{\beta+z}-1)^{|\omega|}}{(e^{\beta}-1)^{|\omega|}}\Big)).
\end{equation*}
We will write \(G_N(z)=\FKlaw_{N,\beta}^{\mathrm{per}}\Big( \frac{(e^{\beta+z}-1)^{|\omega|}}{(e^{\beta}-1)^{|\omega|}}\Big)\).

The goal of this section is the proof of
\begin{lemma}
	\label{lem:F_analytic}
	For any \(\beta,q\) such that hypotheses~\ref{hyp:exp_weak_mix} and~\ref{hyp:exp_dec_finite_connexions} hold, there exists \(\epsilon>0\) such that \(F_{\beta}(z)\) exists and is analytic in the domain \(\{|z|<\epsilon\} \).
\end{lemma}
Which directly implies Theorem~\ref{thm:main}.

\subsection{Random cluster model pressure and associated polymer models}

In this section, we re-write \(G_N(z)=\FKlaw_{N,\beta}^{\mathrm{per}}\Big( \frac{(e^{\beta+z}-1)^{|\omega|}}{(e^{\beta}-1)^{|\omega|}}\Big)\) as the partition function of a polymer model, suitable for the use of cluster expansion.

Set \(\alpha_{z} = \frac{e^{\beta+z}-1}{e^{\beta}-1}\). We now develop a polymer representation for \(\FKlaw_{\beta}\big( \alpha_z^{|\omega|}\big)\) (we keep the \(\mathrm{per}\) and the \(N\) implicit and drop them from the notation from now on). Start by taking a scale \(L\in\{0,1,2,3,\cdots\}\) (we will always consider \(N\) such that \(2N+1\) is divisible by \(2L+1\)). For \(x\in\bbT_N^L\), denote
\begin{equation*}
	f_x(\omega) = \alpha_{z}^{|\omega\cap E_L(x)|}-1.
\end{equation*}

First perform a ``\(+1-1\)'' expansion:
\begin{equation}
\label{eq:CE_Mayer}
	\FKlaw_{\beta}\big( \alpha_z^{|\omega|}\big) = \FKlaw_{\beta}\Big( \prod_{x\in \bbT_N^L}(f_x(\omega) +1) \Big) = \sum_{A\subset \bbT_N^L} \FKlaw_{\beta}\Big( \prod_{x\in A}f_x(\omega) \Big).
\end{equation}
We then use the measure \(P\) constructed in~\ref{thm:dep_encoding_coupling}. Recall that \(\Gamma_N^L\) is the set of connected sub-graphs of \(\bbT_N^{L}\). Denote \(\Gamma_N^L(x)\) the subset of \(\Gamma_N^L\) containing all polymers (connected graphs) containing \(x\). Then,
\begin{equation*}
	\FKlaw_{\beta}\Big( \prod_{x\in A}f_x(\omega) \Big) = \sum_{\gamma_x, x\in A }P\Big( \prod_{x\in A}f_x(\omega)\mathds{1}_{C_x=\tilde{\gamma}_x} \Big),
\end{equation*}where the sum is over \(\tilde{\gamma}_x\in\Gamma_{N}^L(x)\).

To each collection \(\tilde{\gamma}_x, x\in A\), one associate the collection of the connected components of the union \(\bigcup_A \tilde{\gamma}_x\), denoted \(\bar{\gamma}\). This is a set of disjoints polymers. One has
\begin{equation}
\label{eq:CE_facto}
	P\Big( \prod_{x\in A}f_x(\omega)\mathds{1}_{C_x=\tilde{\gamma}_x} \Big) = \prod_{\gamma\in \bar{\gamma}} P\Big( \prod_{x\in A\cap \gamma}f_x(\omega)\mathds{1}_{C_x=\tilde{\gamma}_x} \Big).
\end{equation}
One then expand~\eqref{eq:CE_Mayer} by first summing over \(\bar{\gamma}\) and then summing over compatible \(A\)s and collections \(\gamma_x,x\in A\). Using~\eqref{eq:CE_facto}, one obtains
\begin{multline}
\label{eq:CE_polymers}
	\FKlaw_{\beta}\Big( \prod_{x\in A}f_x(\omega) \Big) = \sum_{\bar{\gamma}}\sum_{A\subset \cup_{\bar{\gamma}}} \sum_{\{\gamma_x\}_{x\in A}\sim \bar{\gamma} } \prod_{\gamma\in \bar{\gamma}} P\Big( \prod_{x\in A\cap \gamma}f_x(\omega)\mathds{1}_{C_x=\tilde{\gamma}_x} \Big)=\\
	\sum_{\bar{\gamma}\subset \Gamma_{N}^L} \prod_{\{\gamma,\gamma'\}\subset \bar{\gamma}}\delta(\gamma,\gamma') \prod_{\gamma\in \bar{\gamma}} \sum_{A\subset \gamma} \sum_{\substack{\{\tilde{\gamma}_x\}_{x\in A}\\ \cup\tilde{\gamma}_x = \gamma} }  P\Big( \prod_{x\in A}f_x(\omega)\mathds{1}_{C_x=\tilde{\gamma}_x} \Big)\equiv\\
	\equiv \sum_{\bar{\gamma}\subset \Gamma_{N}^L} \prod_{\{\gamma,\gamma'\}\subset \bar{\gamma}}\delta(\gamma,\gamma') \prod_{\gamma\in \bar{\gamma}}w_z(\gamma),
\end{multline}where \(\delta(\gamma,\gamma')\) is \(0\) if \(\gamma\cup\gamma'\) is connected and \(0\) otherwise. We succeeded in writing \(G_N(z)\) is the form of a polymer model.

\subsection{Convergence of cluster expansion and analyticity}
We want to use the identity~\eqref{eq:ClusterExpansion}. But to do so, we need to guarantee convergence of the expansion. This is done via Theorem~\ref{thm:ClusterExp_Convergence} and the next lemma.
\begin{lemma}
	\label{lem:upper_bnd_weights}
	Under the hypotheses of Theorem~\ref{thm:dep_encoding_coupling}, there exists \(L\geq 0\), and \(\epsilon>0\) such that for any \(\gamma'\in\Gamma_{N}^L\) and \(|z|<\epsilon\),
	\begin{equation}
		\sum_{\gamma\in\Gamma_N^L} |w_z(\gamma)| e^{4d|\gamma|}|\delta(\gamma,\gamma')-1|\leq 4d|\gamma'|
	\end{equation} uniformly over \(N\) large enough, where \(w_z(\gamma),\delta(\gamma,\gamma')\) are defined as in~\eqref{eq:CE_polymers}.
\end{lemma}
\begin{proof}
	Let \(c_d\) be such that the number of polymers of size \(k\) containing \(0\) is at most \(c_d^k\). By translational invariance of the weights (which follows from translation invariance of \(\FKlaw_{\beta}\)), and the fact that \(\delta(\gamma,\gamma')-1 =0\) if \(\gamma\cup\gamma'\) is not connected, it is sufficient to show that for some \(L\geq 0\) and all \(|z|<\epsilon\),
	\begin{equation}
		|w_z(\gamma)|e^{4d|\gamma|}c_d^{|\gamma|}\leq 2^{-|\gamma|}.
	\end{equation}
	
	By Corollary~\ref{cor:exp_dec_cluster_set},
	\begin{equation*}
	P( C_A=\gamma)\leq P( |C_A|\geq |\gamma|)\leq e^{-c_L |\gamma|} e^{a_L |A|},
	\end{equation*} for some \(c_L>0,a_L\geq 0\) going to \(\infty\) with \(L\). Let \(L\) be fixed such that \(e^{-c_L}\leq 1/(4e^{4d}c_d)\).
	
	Let then \(\epsilon=\epsilon(L)>0\) be such that \(|\alpha_{z}^{k}-1|\leq e^{-a_L}\) for all \(k\in\{0,1,\cdots,|E_L|\}\) whenever \(|z|\leq \epsilon\).
	
	Observe that in this case, for any \(\omega,x\), \(|f_x(\omega)|\leq e^{-a_L}\). Now, fix a polymer \(\gamma\). We have, for any \(|z|<\epsilon\),
	\begin{equation}
	\label{eq:weight_bnd_decomp}
		|w_z(\gamma)| \leq \sum_{A\subset \gamma} e^{-a_L|A|}  P( C_A=\gamma)
		\leq \sum_{k=0}^{|\gamma|} \binom{|\gamma|}{k} e^{-a_L k} e^{-c_L |\gamma|} e^{a_L k} = 2^{|\gamma|}e^{-c_L |\gamma|}.
	\end{equation} By choice of \(L\), \(2^{|\gamma|}e^{-c_L |\gamma|}\leq (2e^{4d}c_d)^{-|\gamma|}\) which concludes the proof.
\end{proof}
We are now ready to prove Lemma~\ref{lem:F_analytic}.
\begin{proof}[Proof of Lemma~\ref{lem:F_analytic}]
	See \cite[Theorem 5.8]{Friedli+Velenik-2017} for a more detailed version of the same argument.
	
	Let \(L,\epsilon\) be given by Lemma~\ref{lem:upper_bnd_weights}. Using~\eqref{eq:ClusterExpansion}, we obtain that the functions \(G_N(z)\) are non-zero in the disk \(D_{\epsilon} = \{z\in\mathbb{C}:|z|<\epsilon\}\) for all \(N\) large enough. Thus, the functions \(F_{N,\beta}(z)\) form a family of analytic functions that are uniformly bounded on \(D_{\epsilon}\). Moreover, the limit \(N\to \infty\) exists for \(z\in\R\). Vitali Convergence Theorem implies that \(F_{\beta}(z)\) exists and is analytic on \(D_{\epsilon}\).
\end{proof}

\subsection{Sketch of the proof of Theorem~\ref{thm:Correlation_analy}}
\label{sec:sketch_Corr_analy}

The treatment of correlation functions is very similar to the one of the pressure. Recall that for \(A\subset \Ed\) finite, \(g_A(\omega)=\prod_{e\in A} \omega_e\). Then (again, dropping the volume from the notation),
\begin{equation*}
	\FKlaw_{\beta+z}(g_A) = \frac{1}{G_N(z)}\FKlaw_{N,\beta}^{\mathrm{per}}\Big( \frac{(e^{\beta+z}-1)^{|\omega|}}{(e^{\beta}-1)^{|\omega|}} g_A\Big).
\end{equation*} Define then \(\Delta_A\subset \bbT_N^L\) the set of vertices \(x\) such that \(E_L(x)\cap A\neq \varnothing\). One then introduce
\begin{equation*}
	\tilde{f}_x^A(\omega) = \begin{cases}
	f_x(\omega) &\text{ if } x\notin \Delta_A,\\
	\alpha_z^{|\omega\cap E_L(x)|}g_{A\cap E_L(x)} -1 &\text{ if } x\in\Delta_A,
	\end{cases}
\end{equation*} and the modified activities: 
\begin{equation*}
	w_{A,z}(\gamma) = \sum_{B\subset \gamma} \sum_{\substack{\{\tilde{\gamma}_x\}_{x\in B}\\ \cup\tilde{\gamma}_x = \gamma} }  P\Big( \prod_{x\in B}\tilde{f}_x^A(\omega)\mathds{1}_{C_x=\tilde{\gamma}_x} \Big),
\end{equation*}
so that
\begin{equation*}
	\FKlaw_{\beta+z}(g_A) = \frac{1}{G_N(z)} \sum_{\bar{\gamma}\subset \Gamma_{N}^L} \prod_{\{\gamma,\gamma'\}\subset \bar{\gamma}}\delta(\gamma,\gamma') \prod_{\gamma\in \bar{\gamma}}w_{A,z}(\gamma).
\end{equation*}
Then one proceed as for the pressure case (notice that the radius of convergence for the polymer model with modified activities depends on \(A\)). Using then~\eqref{eq:ClusterExpansion}, one gets
\begin{equation*}
	\FKlaw_{\beta+z}(g_A) = e^{\tilde{\sum} \cdots - \sum \cdots}
\end{equation*}with \(\tilde{\sum}\) the sum with modified activities. Both sums converge absolutely. The limit \(N\to\infty\) follows the same path as the proof of Lemma~\ref{lem:F_analytic}.

\section*{Acknowledgements}

The author thanks the university Roma Tre for its hospitality and is supported by the Swiss NSF through an early PostDoc.Mobility Grant. 

\appendix
\section{Cluster expansion}
\label{app:ClusterExp}

We recall here what is the cluster expansion of a pair interaction polymer model and a result about convergence of this expansion. The whole presentation can be found in \cite{Friedli+Velenik-2017} so we only state the results and refer to \cite[Chapter 5]{Friedli+Velenik-2017} for proofs and more details. This appendix is the one of~\cite{Ott-2019}, we include it here for the reader convenience.

\subsection*{The Framework}
Suppose we are given a set \(\Gamma\) (the set of polymers), a weighting \(w:\Gamma\to\mathbb{C}\) and an interaction \(\delta:\Gamma\times\Gamma\to[-1,1] \). The \emph{polymer partition function} is then given by
\begin{equation*}
Z = \sum_{H \subset \Gamma \text{ finite}} \Big(\prod_{\gamma\in H}w(\gamma)\Big) \Big(\prod_{\{\gamma,\gamma'\}\subset H} \delta(\gamma,\gamma')\Big).
\end{equation*}
The empty set contributes \(1\) to the sum. To state the formal equality, we need to define the \emph{Ursell function} of an ordered collection of polymers:
\begin{equation*}
U(\gamma_1) = 1,\quad U(\gamma_1,...,\gamma_n) = \frac{1}{n!} \sum_{\substack{G\subset K_n\\ \textnormal{connected}}} \prod_{\{i,j\}\in E_G} (\delta(\gamma_i,\gamma_j)-1),
\end{equation*}where \(K_{n}\) is the complete graph on \(\{1,...,n\}\), \(G=(\{1,...,n\},E_G)\) is an edge-subgraph of \(K_n\).
\subsection*{The Formal Equality} Equipped with this set-up, we have the equality (valid when the sum in the exponential is absolutely convergent)
\begin{equation}
\label{eq:ClusterExpansion}
Z = \exp(\sum_{n\geq 1} \sum_{\gamma_1}...\sum_{\gamma_n} U(\gamma_1,...,\gamma_n) \prod_{i=1}^n w(\gamma_i)).
\end{equation}
\subsection*{Convergence}
The result we will use is the following criterion for the absolute convergence of \(\sum_{n\geq 1} \sum_{\gamma_1}...\sum_{\gamma_n} U(\gamma_1,...,\gamma_n) \prod_{i=1}^n w(\gamma_i)\):
\begin{theorem}
	\label{thm:ClusterExp_Convergence}
	If there exists \(g:\Gamma\to\R_{>0}\) such that for every \(\gamma'\in\Gamma\)
	\begin{equation*}
	\sum_{\gamma\in\Gamma} e^{g(\gamma)} |w(\gamma)| |\delta(\gamma,\gamma')-1|\leq g(\gamma'),
	\end{equation*}
	and such that \(\sum_{\gamma\in\Gamma} e^{g(\gamma)}|w(\gamma)|<\infty\) then,
	\begin{equation*}
	\sum_{n\geq 1} \sum_{\gamma_1}...\sum_{\gamma_n} |U(\gamma_1,...,\gamma_n)| \prod_{i=1}^n |w(\gamma_i)|<\infty.
	\end{equation*}
\end{theorem}

\bibliographystyle{plain}
\bibliography{BIGbib}

\begin{thebibliography}{10}

\bibitem{Aizenman+Barsky-1987}
M.~Aizenman and D.~J. Barsky.
\newblock Sharpness of the phase transition in percolation models.
\newblock {\em Communications in Mathematical Physics}, 108(3):489--526, 1987.

\bibitem{Aizenman+Barsky+Fernandez-1987}
M.~Aizenman, D.~J. Barsky, and R.~Fern\'andez.
\newblock The phase transition in a general class of {I}sing-type models is
  sharp.
\newblock {\em J. Statist. Phys.}, 47(3-4):343--374, 1987.

\bibitem{Bodineau-2005}
T.~Bodineau.
\newblock Slab percolation for the {I}sing model.
\newblock {\em Prob. Th. Rel. Fields}, 132(1):83--118, 2005.

\bibitem{Dobrushin+Shlosman-1987}
Shlosman~S.B. Dobrushin, R.L.
\newblock Completely analytical interactions: Constructive description.
\newblock {\em J. Stat. Phys.}, 46:983–1014, 1987.

\bibitem{Duminil-Copin+Goswami+Raoufi-2019}
H.~Duminil-Copin, S.~Goswami, and A.~Raoufi.
\newblock Exponential decay of truncated correlations for the {Ising} model in
  any dimension for all but the critical temperature.
\newblock {\em Commun. Math. Phys.}, October 2019.

\bibitem{Duminil-Copin+Raoufi+Tassion-2017}
H.~Duminil-Copin, A.~Raoufi, and V.~Tassion.
\newblock Sharp phase transition for the random-cluster and {P}otts models via
  decision trees.
\newblock {\em Annals of Mathematics}, 189:75--99, 2019.

\bibitem{Duminil-Copin+Tassion-2016}
H.~Duminil-Copin and V.~Tassion.
\newblock A new proof of the sharpness of the phase transition for {Bernoulli}
  percolation and the {Ising} model.
\newblock {\em Commun. Math. Phys.}, 343:725–745, 2016.

\bibitem{Friedli+Velenik-2017}
S.~Friedli and Y.~Velenik.
\newblock {\em Statistical Mechanics of Lattice Systems: A Concrete
  Mathematical Introduction}.
\newblock Cambridge University Press, 2017.

\bibitem{Grimmett-1999}
G.~Grimmett.
\newblock {\em Percolation}.
\newblock Springer, Berlin Heidelberg, 1999.

\bibitem{Harel+Spinka-2018}
M.~Harel and Y.~Spinka.
\newblock {F}initary codings for the random-cluster model and other
  infinite-range monotone models.
\newblock 2018.

\bibitem{Isakov-1984}
S.~N. Isakov.
\newblock Nonanalytic features of the first order phase transition in the
  {I}sing model.
\newblock {\em Comm. Math. Phys.}, 95(4):427--443, 1984.

\bibitem{Liggett+Schonmann+Stacey-1997}
T.~M. Liggett, R.~H. Schonmann, and A.~M. Stacey.
\newblock Domination by product measures.
\newblock {\em The Annals of Probability}, 25(1):71--95, January 1997.

\bibitem{Martinelli-1999}
F.~Martinelli.
\newblock {\em Lectures on Glauber Dynamics for Discrete Spin Models}.
\newblock Springer Berlin Heidelberg, Berlin, Heidelberg, 1999.

\bibitem{Martinelli+Olivieri-1994}
F.~Martinelli and E.~Olivieri.
\newblock Approach to equilibrium of {Glauber} dynamics in the one phase
  region. {I}. {The} attractive case.
\newblock {\em Comm. Math. Phys.}, 161(3):447--486, 1994.

\bibitem{Martinelli+Olivieri+Schonmann-1994}
F.~Martinelli, E.~Olivieri, and R.~H. Schonmann.
\newblock For 2-{D} lattice spin systems weak mixing implies strong mixing.
\newblock {\em Comm. Math. Phys.}, 165(1):33--47, October 1994.

\bibitem{Menshikov-1986}
M.~V. Menshikov.
\newblock Coincidence of critical points in percolation problems.
\newblock {\em Soviet Mathematics Doklady}, 33:856--859, 1986.

\bibitem{Ott-2019}
S.~Ott.
\newblock Weak mixing and analyticity of the pressure in the {I}sing model.
\newblock {\em Comm. Math. Phys.}, Oct 2019.

\bibitem{Pisztora-1996}
A.~Pisztora.
\newblock Surface order large deviations for {I}sing, {P}otts and percolation
  models.
\newblock {\em Probab. Th. Rel. Fields}.

\bibitem{Schonmann-1994}
R.~H. Schonmann.
\newblock Slow droplet-driven relaxation of stochastic {Ising} models in the
  vicinity of the phase coexistence region.
\newblock {\em Comm. Math. Phys.}, 161(1):1--49, 1994.

\bibitem{Schonmann+Shlosman-1995}
R.~H. Schonmann and S.~B. Shlosman.
\newblock Complete analyticity for 2d {Ising} completed.
\newblock {\em Comm. Math. Phys.}, 170(2):453--482, June 1995.

\bibitem{van_Enter+Fernandez+Schonmann+Shlosman-1997}
A.~C.~D. van Enter, R.~Fernández, R.~H. Schonmann, and S.~B. Shlosman.
\newblock Complete {Analyticity} of the 2d {Potts} {Model} above the {Critical}
  {Temperature}.
\newblock {\em Commun. Math. Phys.}, 189(2):373--393, November 1997.

\end{thebibliography}

\end{document}